\documentclass[letter,10pt]{article}
\usepackage{amssymb}
\usepackage{amsthm}
\theoremstyle{dotless}
\usepackage{amsmath}
\usepackage{mathrsfs}
\usepackage{verbatim}
\usepackage{enumerate}

\usepackage{cite}
\usepackage{hyperref}
\usepackage{changes}
\usepackage{color}
\usepackage[toc,page, title]{appendix}

\newtheorem{theorem}{Theorem}[section]
\newtheorem{lemma}[theorem]{Lemma}
\newtheorem{corollary}[theorem]{Corollary}
\newtheorem{example}[theorem]{Example}
\newtheorem{proposition}[theorem]{Proposition}
\newtheorem{definition}[theorem]{Definition}
\newtheorem{remark}[theorem]{Remark}
\newcommand\argmax{{\rm argmax}}

\begin{document}
\title{Interpolating between matching and hedonic pricing models\footnote{The author is pleased to acknowledge the support of a University of Alberta start-up grant and National Sciences and Engineering Research Council of Canada Discovery Grant number 412779-2012. 
}}

\author{Brendan Pass\footnote{Department of Mathematical and Statistical Sciences, 632 CAB, University of Alberta, Edmonton, Alberta, Canada, T6G 2G1 pass@ualberta.ca.}}
\maketitle
\begin{abstract} 
We consider the theoretical properties of a model which encompasses bi-partite matching under transferable utility on the one hand, and hedonic pricing on the other.  This framework is intimately connected to tripartite matching problems (known as multi-marginal optimal transport problems in the mathematical literature).  We exploit this relationship in two ways; first, we  show that a known structural result from multi-marginal optimal transport can be used to establish an upper bound on the dimension of the support of stable matchings.  Next, assuming the distribution of  agents on one side of the market is  continuous, we identify a condition on their preferences that ensures purity and uniqueness of the stable matching; this condition is a variant of a known condition in the mathematical literature, which guarantees analogous properties in the multi-marginal optimal transport problem.   We exhibit several examples of surplus functions for which our condition is satisfied, as well as some for which it fails.
\end{abstract}

\section{Introduction}

This paper considers the theoretical properties of a general model, which encompasses both the transferable utility matching model of Shapley and Shubik \cite{SS} and Becker \cite{Becker73}, extended to continuous type spaces by Gretsky, Ostroy and Zame \cite{GOZ}, and the hedonic model of Rosen \cite{Rosen}, whose theoretical properties (in a continuous,  multi-dimensional setting)  were studied by Ekeland \cite{ekeland3,ekeland2} and Chiappori-McCann-Nesheim \cite{CMN}.

In the hedonic model, agents on two sides of a market (eg, buyers and sellers) are matched together according to  their preferences to exchange certain goods, assuming they are indifferent to which partner they do business with.  Letting $X$ and $Y$ be spaces of buyers and sellers, respectively, and $Z$ the set of goods that can feasibly be produced, we assume that agents' preferences are encoded respectively by functions $u(x,z)$ and $v(y,z)$, expressing the utilities of buyer $x \in X$ and seller $y \in Y$ if they purchase or produce a product of type $z \in Z$, respectively.    The main problem is then to determine which buyers match with which sellers,  which goods they exchange and the prices they exchange for them, in equilibirum.

In the matching model, agent $x$ (respectively $y$) has a preference $u(x,y)$ (respectively $v(x,y)$) to match with agent $y$ (respectively $x$).  Together, a match between $x$ and $y$  generates a \emph{joint} utility, or surplus, $s(x,y)=u(x,y) +v(x,y)$.  Utility can then be \emph{transferred} in the form of a payment from one agent to the other; by this mechanism, the total surplus $s(x,y)$ can be divided in any way between the two agents.  Here the good $z$ to be exchanged (or the non monetary terms of the contract) does not affect agents' preferences.

Recently, Dupuy, Galichon and Zhao \cite{DGZ} formulated a hybrid model in which agents have preferences which depend on \emph{both}  their partners and the product under exchange.     In that paper, $x$ and $y$ represent agents on different sides of a marriage market, and $z$ the location where they would agree to settle; however, as was noted by the authors, the model has much wider potential applicability.  In family economics, for instance, when a couple match together, there are many conditions of the match which can affect the surplus generated.  In addition to location, a couple may choose to marry or to live together without marrying, and to have one or more children, for example. These decisions affect the surplus differently depending on the couple; see, for example, \cite{MourifieSiow14}.
When considering buyers and sellers, it also seems natural in many scenarios to allow consumers' preferences to depend on both the producer he does business with and the good he receives.  For example, consumers often exhibit \textit{brand loyalty}; they may be willing to pay more (for the same good) when dealing with one company rather than another.  On the other hand, producers' preferences can also depend on the consumers they sell to; for example, mortgage lenders often offer better rates to clients with higher credit scores, reflecting the fact that they prefer to do business with more credit worthy borrowers.  Another example occurs in insurance models like the one of Rothschild and Stiglitz \cite{RothschildStiglitz}; insurance companies offer the same policy at different prices to different consumers depending on how their characteristics influence the risk of a claim.  

The theortical properties of these hybrid models do not seem to have received much attention.  In this paper, we study this model when the type spaces $X$ and $Y$ and the space of feasible contracts $Z$ are all continuous.  In both the classical matching problem, and the hedonic problem, conditions are well known which ensure that the equilibrium, or stable matching, is unique and pure.  In the matching problem, this condition is  known as the \textit{twist}, or \textit{generalized Spence-Mirrlees} condition.   One natural  question in the present setting is to identify conditions on the joint surplus function (which is now a function of $x,y$ and $z$, without any specific form) in the more general model which ensure purity and uniqueness of the stable match.  

Both the strict matching and hedonic problems have well established connections to a variational problem known as \emph{optimal transportation}; economically, this is exactly the social planner's problem of matching the agents in order to maximize their average surplus (a detailed introduction to this subject can be found in Galichon \cite{Galichon16}, Santambrogio \cite{Santambrogio15p} and Villani \cite{V2}).  The generalized model studied here turns out to have natural connections to both the classical optimal transport problem, and variant of it where there are three (rather than two) measures to be matched (known in the mathematics literature as a multi-marginal optimal transport problem).    Our main contribution here is to establish and exploit these connections to uncover insights into matching patterns for this generalized matching-hedonic model.  As an immediate application of the multi-marginal optimal transport point of view, we establish an upper bound on the dimension of the support of stable matchings (which are measures on the product space $X \times Y \times Z$) in terms of the signature of the off-diagonal part of the Hessian of $s$.   We then find conditions on the functions $u$ and $v$, and the measures $\mu$ and $\nu$ ensuring existence, uniqueness and purity of equilibria.   Our condition here is a weaker variant of the  \textit{twist on splitting sets} condition, which is known to ensure purity and uniqueness  in the multi-marginal optimal transport problem; we will call the condition for the hybrid matching problem the \textit{twist on $z$-trivial splitting sets}.


In addition, due to recent work of Chiappori, McCann and Nesheim \cite{CMN}, it is now clear that the hedonic problem is actually equivalent to a matching problem, with a surplus  equal to the the maximum possible joint utility for buyers and sellers, among all possible goods; we often refer to the analogous maximized surplus $\bar s(x,y)$ (see \eqref{redsur}) in our setting as the \emph{reduced surplus}, as we have reduced the number of variables from three to two (that is, $\bar s$ depends only on $x$ and $y$).   This equivalence persists in our setting. 
 In this simplified but equivalent bipartite matching setting, the twist condition on the maximal surplus \eqref{redsur} ensures the uniqueness and purity of the stable match.   It is desirable then, to understand when the twist condition on the surplus \eqref{redsur} holds; we show that, under certain conditions, this is essentially equivalent to our twist on $z$-trivial  splitting sets condition on $ s$.  We believe that this equivalence indicates that twist on $z$-trivial splitting sets is a natural condition for the hybrid problem.

In the next section, we outline the model under consideration and establish some of its basic properties.  In the third section, we use the connection with multi-marginal optimal transport to establish a result on the local structure (ie, the dimension of the support) of stable matchings.   In section four we develop  our sufficient condition for purity and uniqueness of stable matchings, while section five is devoted to the reformulation of our problem as a strict matching problem (with a reduced surplus function) and the demonstration that the classical twist condition of $\bar s$ is equivalent to our twist on $z$-trivial splitting sets condition on the surplus $s$.   The sixth section presents some examples, while we offer a brief conclusion in the final section.

Short, simple proofs of mathematical results are included within the body of the paper; longer or more involved mathematical arguments are relegated to an appendix, to avoid interupting the flow of the paper.

\section{The general model and basic properties}\label{sect: model}
We consider heterogeneous distributions of buyer and seller types, encoded respectively by compactly supported  Borel probability measures $\mu$ on $X \subseteq \mathbb{R}^{n_x}$ and $\nu$ on $Y \subseteq \mathbb{R}^{n_y}$, and a set of feasible goods, parameterized by $Z \subseteq \mathbb{R}^{n_z}$. Tthe sets $X$, $Y$ and $Z$ will be the closures of open and bounded sets $X^0,Y^0$ and $Z^0$, respectively, with smooth boundaries.  Each buyer will  buy exactly one good; each seller will produce and sell exactly one good.\footnote{It would be straightforward to enhance the model to allow for unequal numbers of buyers and sellers, and to allow both to decline to participate in any match; this can be done as in \cite{CMN}, by augmenting the measures $\mu$ and $\nu$ with Dirac masses, representing null buyers and sellers.  As this is tangential to our main purpose here, we work instead with the simpler model in which all agents participate.}

The preference of a buyer of type $x$ to purchase a good of type $z$ from a seller of type $y$ will be given by a function $u(x,y,z)$, while the preference of a seller of type $y$ to sell a good of type $z$ to a buyer of type $x$ is given by $v(x,y,z)$.  We will assume throughout the paper that $u$ and $v$ are uniformly Lipschitz; stonger regularity hypotheses will be adopted at various specific points.  
Utilities will be quasilinear, so that the utility derived by a buyer purchasing a good of type $z$ from a seller of type $y$ for a price $p$ will be 

$$
u(x,y,z) -p
$$
and similarly, the  utility derived by a seller selling a good of type $z$ to a buyer of type $x$ for a price $p$ will be 

$$
v(x,y,z)+p.
$$
We will denote by $s(x,y,z)$ the total, or joint, surplus generated when a buyer of type $x$ purchases a good of type $z$ from a seller of type $y$:
$$ 
s(x,y,z) =u(x,y,z) +v(x,y,z).
$$

As the utility can freely be transferred from one partner to another, the analytical properties of $s$, rather than $u$ and $v$ separately, are most relevant in determining the purity and uniqueness of equilibrium.


We define a \emph{matching} as a probability measure $\gamma$ on $X \times Y \times Z$ whose first marginal is $\mu$ and whose second marginal is $\nu$; that is
$$
\gamma(A \times Y \times Z) =\mu(A), \text{  }\gamma(X \times B \times Z) =\nu(B)
$$
for all Borel $A \subseteq X$, $B \subseteq Y$.  This represents an assignment of the agents in the sets $X$ and $Y$  into pairs, and assigns to each pair a good from the set $Z$ to be exchanged.  We will denote by $\Gamma_{XYZ}(\mu,\nu)$ the set of all matchings  of $\mu$ and $\nu$ on $X \times Y \times Z$.

We will say that a mapping $T:X \rightarrow Y$ \textit{pushes $\mu$ forward to $\nu$}, and write $\nu =T_{\#}\mu$ if $\mu(T^{-1}(B)) =\nu(B)$ for all Borel $B \subseteq Y$.

 For a given measure $\gamma$ on $X \times Y \times Z$, we denote by $\gamma_{XY}$  its projection onto $X \times Y$, a measure on $X \times Y$ defined by $\gamma_{XY}(C) =\gamma(C \times Z)$ for any $C \subset X \times Y$.  The measures $\gamma_{XZ}$ and $\gamma_{YZ}$  are defined analogously.  Note that $\gamma_{XY} = (P_{XY})_\# \gamma$, where $P_{XY}: X \times Y \times Z \rightarrow X\times Y$ is the projection map, $P_{XY}(x,y,z) :=(x,y)$.  It is also worth noting that if $\gamma$ is a matching of $\mu$ and $\nu$, then $\gamma_{XY}$ has marginals $\mu$ and $\nu$.

Given a matching $\gamma$, functions $U(x)$ and $V(y)$ are called \emph{payoff} functions for $\gamma$ if 
$$
U(x)+V(y) = s(x,y,z)
$$
for $\gamma$ almost every $(x,y,z)$.  For any matching, points in the \emph{support}\footnote{The support of $\gamma$  is the smallest closed set $spt(\gamma) \subset X \times Y \times Z$ with full mass, $\gamma(spt(\gamma))=1$. } $spt(\gamma)$ of $\gamma$ (and hence in the equality set $\{U(x)+V(y) = u(x,y,z) + v(x,y,z)\}$ for payoff functions $U$ and $V$) represent buyer-seller pairs who are matched together by $\gamma$, together with the good they exchange.  Payoff functions then represent a division of the surplus between matched pairs.  Given such a triple $(x,y,z) \in spt(\gamma)$ and payoffs $U(x)$ and $V(y)$, the price\footnote{In contrast to the strict hedonic problem, one cannot hope for a market clearing pricing function $p(z)$ which is independent of $x$ and $y$ here; it is possible that  the same good may be exchanged between different pairs of buyers and seller for different prices.}  that $y$ charges $x$ in exchange for the good $z$ is given 

\begin{equation}
 p_{x,y,z}:=u(x,y,z)-U(x)=V(y)-v(x,y,z) .
\end{equation}

 The matching is called \emph{stable} if there exist payoff functions  $U(x)$ and $V(y)$ such that


\begin{equation}\label{stabilitycond}
U(x)+V(y) \geq u(x,y,z) + v(x,y,z)
\end{equation}
for all $(x,y,z)$.


The condition \eqref{stabilitycond} ensures \emph{stability} of the matching  in the sense that no  pair of unmatched agents would both prefer to leave their current partners and match together.  If   \eqref{stabilitycond} failed, so that $U(x)+V(y) < u(x,y,z) + v(x,y,z)$ for some unmatched buyer-seller-good triple $(x,y,z)$ (that is, $(x,y,z) \notin spt(\gamma)$), then buyer $x$ and seller $y$ would be incentivized to exchange good $z$ for a price $p$ such that
$$
V(y)- v(x,y,z)< p < u(x,y,z)-U(x),
$$ 
resulting in\textit{ increased} payoffs $\bar U(x) := u(x,y,z)-p >U(x)$ and $\bar V(y) := v(x,y,z)+p >V(y)$  for both $x$ and $y$.

Finally, we turn our attention to purity of matchings.  There are several relevant notions of purity here, corresponding to various relationships between buyers, sellers, and goods.
\begin{definition} A matching $\gamma$ is called \emph{buyer-seller pure} if its projection $\gamma_{XY}$ onto $X \times Y$ is concentrated on a graph over $X$;  that is, if there exists a function $F_Y: X \rightarrow Y$ such that $\gamma_{XY}=(Id,F_Y)_\# \mu$.   We say $\gamma$ is \emph{buyer-good pure} if its projection $\gamma_{XZ}$ onto $X \times Z$ is concentrated on a graph over $X$;  that is, if there exists a function $F_Z: X \rightarrow Z$ such that $\gamma_{XZ}=(Id,F_Z)_\# \mu$.  

We will call $\gamma$ \emph{buyer-(seller, good) pure} (or simply \emph{pure}) if it is both buyer-seller and buyer-good pure, which means that $\gamma$ is concentrated on  a graph over $X$.  In other words, there exist functions $F_Y: X \rightarrow Y$ and $F_Z: X \rightarrow Z$ such that $\gamma=(Id,F_Y,F_Z)_\# \mu$.

\end{definition}
Note that one could analogously define several other notions of purity (seller-buyer, good-seller, etc).   The economic interpretation of, for instance, buyer-seller purity is that there is no randomness in each buyer $x$'s choices of the seller $y=F_Y(x)$ he works with; buyers of the same type will (almost) always buy goods from sellers of the same type.   

One of our main contributions in this paper is to identify a condition on the surplus that ensures full, buyer-(seller, good) purity; as we will see, the same condition will guarantee uniqueness of the stable matching as well.
\subsection{Variational interpretation}\label{sect: variational interpretation}

Much like the classical matching and hedonic problems, the problem of finding stable matchings in our setting has a variational formulation.  
Consider the problem of maximizing 
\begin{equation}\label{mm}
\int_{X \times Y \times Z} s(x,y,z) d\gamma(x,y,z)
\end{equation}
over the set $\Gamma_{XYZ}(\mu,\nu)$  of all matchings of $\mu$ and $\nu$ (that is, maximizing the total surplus of all agents).

\begin{theorem}\label{varformulation}
A matching $\gamma$ is an equilibrium if and only if it is optimal in \eqref{mm}.
\end{theorem}
This result is well known in the classical matching case in \cite{GOZ}, when the surplus $s$ (and hence the matching measures $\gamma$ as well) depends only on $x$ and $y$.  For general hybrid surplus functions, $s(x,y,z)$, the result is proven in the discrete case in \cite{DGZ}.  The proof here requires no new ideas, but is included in an appendix in the interest of completeness.  

By standard arguments, Theorem \ref{varformulation} implies existence of a stable matching.

\begin{corollary}
There exists at least one stable matching $\gamma$.
\end{corollary}

\begin{proof}
The proof is by continuity and compactness, and is completely standard.

Continuity of $s$ immediately implies continuity of 
$$
\gamma \mapsto \int_{X \times Y \times Z} s(x,y,z) d\gamma(x,y,z)
$$ with respect to weak convergence of measures.  The Riesz-Markov theorem identifies the dual of the set $C( X \times  Y \times  Z)$ of continuous functions on $X \times  Y \times  Z$ with the set  $\mathcal M(  X \times  Y \times  Z)$ of regular Borel measures on $X \times  Y \times  Z$, with norm given by total variation (which is total mass, for positive measures).  The Banach-Alaoglu theorem then asserts that the closed unit ball in $\mathcal M(X \times  Y \times  Z)$ is  compact.  The set $\Gamma_{XYZ}(\mu,\nu)$ is clearly a weakly closed subset of this unit ball, and therefore is itself compact.  The existence of a maximizer of \eqref{mm} over the set $\Gamma_{XYZ}(\mu,\nu)$, and hence a stable matching by Theorem \ref{varformulation}, then follows immediately.
\end{proof}
\subsection{Connection to tripartite matching}
Problem \eqref{mm} is closely related to a \textit{tripartite} matching problem (also known, in the mathematical literature, as the multi-marginal optimal transport problem), where in addition to prescribing the distributions of agents  $\mu$ on $X$, $\nu$ on $Y$, one  fixes the distribution $\alpha$ on $Z$ \footnote{In these tri-partitie matching problems, the variables are typically interpreted differently; economically, they model problems where three agents are required to form a match (think, for example, of firms hiring simultaneously both CEOS and CFOS, drawn from separate distributions).  The distributions of all three types of agents (firms, CEOs and CFOs)  are then known, and finding a stable match is equivalent to maximizing \eqref{classmm} (see Carlier and Ekeland \cite{ce}). }.   Finding a stable matching in this  problem is equivalent to the following maximization:
\begin{equation}\label{classmm}
T(\mu,\nu,\alpha):=\max_{\gamma \in \Gamma_{XYZ}(\mu,\nu,\alpha)}\int_{X \times Y \times Z}s(x,y,z)d\gamma(x,y,z)
\end{equation}
 where the maximum is over the set $\Gamma_{XYZ}(\mu,\nu,\alpha)$ of positive measures on $X \times Y \times Z$ whose marginals are $\mu, \nu$ and $\alpha$.   The underlying  relationship between tripartite matching and our present   problem is that  the variational problem \eqref{mm} (and therefore the equivalent hybrid matching-hedonic  problem) is equivalent to maximizing $T(\mu,\nu,\alpha)$ over the set of all probability measures $\alpha$ on $Z$.


There is a growing mathematical and economic literature on tripartite (or, more generally, multipartite) matching, which will be useful in what follows, as some of the results there can be translated to the present setting.   In particular, an immediate application is an upper bound on the dimension on the support of the stable matching, which is presented in the next section.
 In addition, conditions ensuring purity and uniqueness in \eqref{classmm} have been identified in \cite{KP2}.  In a subsequent section, we use this as a guide to develop a similar condition for problem \eqref{mm}; our condition here is somewhat weaker than the one in \cite{KP2}, as we do not require purity in \eqref{classmm} for \emph{every} choice of $\alpha$; we require it  only for those which \emph{maximize} $\alpha \mapsto T(\mu,\nu,\alpha)$.

\section{Dimension of the support of  matching measures}
Even when the conditions for purity and uniqueness developed in the next section fail, there are results known about the local structure of the optimizer in \eqref{classmm}; as any stable matching $\gamma$ maximizes \eqref{classmm}, taking $\alpha$ to be its $z$ marginal, these results immediately apply to stable matchings as well.

More specifically, for a $C^2$ surplus function, the theorem below provides a bound on the Hausdorff dimension of the support of $\gamma$ in terms of the off diagonal part of the Hessian of $s$.
Consider the symmetric $(n_x+n_y+n_z) \times (n_x+n_y+n_z) $ matrix
\begin{equation*} \qquad
G:=
\begin{bmatrix}
0 & D^2_{xy}s & D^2_{xz}s\\
 D^2_{yx}s& 0 &  D^2_{yz}s\\
 D^2_{zx}s&  D^2_{zy}s & 0\\
\end{bmatrix}
\end{equation*}
where the three diagonal $0$ blocks are $n_x \times n_x$, $n_y \times n_y$ and $n_z \times n_z$, respectively,  $D^2_{xy}s:=\Big ( \frac{\partial ^2 s}{\partial x_i \partial y_j} \Big)_{ij}$ is the $n_x \times n_y$ matrix of mixed second order partial derivatives with respect to the components of $x$ and $y$, and the other non-zero blocks are defined similarly.  

Recall that the \textit{signature} $(\lambda_+,\lambda_-,\lambda_0)$ of a symmetric $N \times N$ matrix is an ordered triple representing the numbers $\lambda_+, \lambda_-$ and $\lambda_0=N-\lambda_+-\lambda_-$  of positive,  negative and  zero eigenvalues, respectively.

\begin{theorem}\label{thm: local structure}
Assume that $s \in C^2(X \times Y \times Z)$ and that at some point $(x_0,y_0.z_0) \in X \times Y \times Z$, the signature of $G$ is $(\lambda_+, \lambda_-, n_x +n_y+n_z-\lambda_+-\lambda_-)$.  Then there is a neighbourhood $U$ of $(x_0,y_0,z_0)$ in $X \times Y \times Z$ such that $spt(\gamma) \cap U$ is contained in a Lipschitz submanifold of $X \times Y \times Z$ of dimension $n_x+n_y+n_z-\lambda_-$.   
\end{theorem} 
The same result is proven for optimizers of \eqref{classmm} in \cite{P, P4}, and that result immediately implies this one.   Note that the dimension $n_x+n_y+n_z-\lambda_-$ is the number of non-negative eigenvalues of $G$; in fact, if  $spt(\gamma)$ is a differentiable manifold at $(x_0,y_0,z_0)$, then $v^TGv \geq 0$ for any $v$ in the tangent space of $spt(\gamma)$ \cite{P,P4}.  It is worth noting that, unlike the purity results in the subsequent section, this theorem does not require any regularity assumptions on the marginals $\mu$ and $\nu$.  

The following proposition, also established in \cite{P4}, asserts that when the dimensions are all equal, the signature can be determined from the symmetric part of the product $ D^2_{zy}s[ D^2_{xy}s]^{-1} D^2_{xz}s$.
\begin{proposition}
 If $n_x=n_y=n_z=:n$, and the matrices  $ D^2_{xy}s,  D^2_{xz}s$ and $ D^2_{yz}s$ are all invertible, then the signature of $G$ is given by $(\lambda_+,\lambda_-, \lambda_0) = (n+r_-,n+r_+,n-r_--r_+)$ where $r_+$ (respectively $r_-$) is the number of positive (respectively negative) eigenvalues of the $n \times n$ symmetric matrix 
$$ D^2_{zy}s[ D^2_{xy}s]^{-1} D^2_{xz}s +D^2_{zx}s[ D^2_{yx}s]^{-1} D^2_{yz}s.$$
\end{proposition}
In particular, if $n_x=n_y=n_z=1$ and  $\frac{\partial^2s }{\partial z \partial y}[\frac{\partial^2s }{\partial x \partial y}]^{-1}\frac{\partial^2s }{\partial x \partial z}>0$, then $(r_+,r_-) =(1,0)$ in the proposition above and so the proposition together with Theorem \ref{thm: local structure} assert that any stable matching is concentrated on a $1$-dimensional Lipschitz submanifold; that is, a curve.  We will see later on that for an absolutely continuous $\mu$, the same condition ensures purity and uniqueness.  

In higher (but still equal) dimensions, the situation is more subtle.  For a bilinear surplus function $s(x,y,z) = x^TAy +x^TBz + y^TCz +f(x) +g(y)+h(z)$, as in Example \ref{bilinear_utilities} below, the condition  
$$
D^2_{zy}s[ D^2_{xy}s]^{-1} D^2_{xz}s +D^2_{zx}s[ D^2_{yx}s]^{-1} D^2_{yz}s=C^TA^{-1}B +B^T(A^T)^{-1}C>0,
$$ 
together with absolute continuity of $\mu$ implies purity (see Example \ref{bilinear_utilities}), but for more general forms of $s$,  one can have solutions which concentrate on $n$ dimensional sets but are not pure.  Consider, for example, the surplus on $\mathbb{R}^2 \times \mathbb{R}^2 \times \mathbb{R}^2$ from \cite{MomeniPass15}

$$
s(x,y,z)=e^{x^1+y^1}\cos(x^2-y^2)+e^{x^1+z^1}\cos(x^2-z^2)+e^{y^1+z^1}\cos(z^2-y^2)-e^{2x^1} -e^{2y^1}-e^{2z^1}.
$$
For this surplus, a straightforward calculation, found in \cite{MomeniPass15}, verifies that the product  $D^2_{zy}s[ D^2_{xy}s]^{-1} D^2_{xz}s +D^2_{zx}s[ D^2_{yx}s]^{-1} D^2_{yz}s$ is a scalar multiple of the identity, and the results above then imply that the signature of $G$ is $(2,4,0)$.  Every stable matching for this surplus therefore  concentrates on sets of no more than $2$ dimensions.  

However, stable matchings may not be pure.  It is straightforward to check that $s(x,y,z) \leq 0$ for all $(x,y,z)$, with equality on the set 
$$
(x,y,z) \in S=\{(x,y,z): x^1=y^1=z^1\text{ and } x^2-y^2 = 2h\pi, \text{ } x^2-z^2 =2l\pi\text{ for some integers }h,l \}.
$$
It follows that any $\gamma$ concentrated on $S$ is stable (we can take $U=V=0$ as the payoff functions)..
This set is two dimensional, as predicted by the calculations above, but not concentrated on a graph, and so the matching is not pure. 


\section{Conditions for purity and uniqueness}
We now turn our attention to the purity and uniqueness of stable matchings.  For the sake of comparison, we first recall known purity and uniqueness results for the simpler, strict matching and hedonic problems.  The \textit{twist}, or \textit{generalized Spence-Mirrlees}, condition plays a fundamental role in that setting:
\begin{definition}
Given a differentiable function, say $s(x,y)$, of two variables, we say $u$ is \emph{$x-y$ twisted} if for each $x \in X$, the mapping 
$$
y \mapsto D_xs(x,y)
$$
is injective.   Here, $D_x$ represents the gradient of $s$ with respect to the $x$ variable.

We will use the same terminology for functions of several variables, when all but one are held fixed.  That is, we say $s(x,y,z)$ is  \emph{$x-z$ twisted} if for each $x \in X, y \in Y$, the mapping 
$$
z \mapsto D_xs(x,y,z)
$$
is injective. 
\end{definition}

\subsection{Classical matching and hedonic problems}
We first review a purity result in the straight matching case, $s=s(x,y)$.

\begin{theorem}\textbf{(Matching problems)}\label{matching_purity}
Assume that $u=u(x,y)$ and $v=v(x,y)$ depend only on $x$ and $y$.  Assume that $\mu$ is absolutely continuous with respect to Lebesgue measure and $s(x,y) =u(x,y) +v(x,y)$  is $x-y$ twisted.  Then any stable matching  is buyer-seller pure and its projection $\gamma_{XY}$ onto $X \times Y$ is uniquely determined; that is, if $\gamma$ and $\bar \gamma$ are stable  matchings, $\gamma_{XY} =\bar\gamma_{XY}$.
\end{theorem}
This result is well known; a proof can be found in \cite{CMN}.  Indeed, in the mathematics literature, comparable results  regarding the equivalent optimal transport problem were established, in various levels of generality,  by Brenier\cite{bren}, Gangbo \cite{G}, Levin \cite{lev}, Gangbo-McCann\cite{gm} Caffarelli \cite{Caf}.

 Note that in our terminology, stable matchings are measures on $X \times Y \times Z$, whereas in the literature the strict matching problem is usually formulated instead  in terms of measures on $X \times Y$, as the good $z$ plays no role in the surplus function.  In our formulation, we would not have full uniqueness; any measure on $X \times Y \times Z$, whose projection onto $X \times Y$ is $\gamma_{XY}$ is a stable matching, as both agents $x$ and $y$ are indifferent to the superfluous good $z$.

  We now turn to the fully hedonic case, where agents' preferences $u=u(x,z)$ and $v=v(y,z)$ depend on goods but not on their partners.   

\begin{theorem}\textbf{(Hedonic problems)}
Assume that agents' preferences $u=u(x,z)$ and $v=v(y,z)$ depend only on $(x,z)$ and $(y,z)$, respectively, and that $\mu$ is absolutely continuous with respect to Lebesgue measure. Then:

\begin{enumerate}
\item If $u$ is $x-z$ twisted, the stable matching is buyer-good pure and it's projection $\gamma_{XZ}$ onto $X \times Z$ is uniquely determined.
\item If in addition, $v$ is $z-y$ twisted, and, for each fixed $x,y$, every maximum of the mapping $z \mapsto u(x,z) +v(y,z)$ over $Z$ occurs on the \emph{interior} of $Z$,  the stable matching measure is buyer-(seller,good) pure and unique.
\end{enumerate}

\end{theorem}
The proof of part 1 can be found in \cite{ekeland2}, while the proof of the second assertion requires a minor additional argument. 

\begin{proof} \textit{(of assertion 2)}


Using part 1), we have the existence of a unique map $F_Z:X \rightarrow Z$ such that, for $\gamma$ almost every $(x,y,z)$, $z=F_Z(x)$.   Now, by a result in \cite{CMN}, we also have for $\gamma$ almost every $(x,y,z)$ that $z$ maximizes $z' \mapsto u(x,z') +v(y,z')$, so that 
$$
D_zu(x,z) =-D_zv(y,z)
$$
or
\begin{equation}\label{envelope for z}
D_zu(x,F_Z(x)) =-D_zv(y,F_Z(x)).
\end{equation}
The $z-y$ twist condition then ensures that there is only one $y$ satisfying this equation. That is, $y=:F_Y(x)$ is uniquely  determined by $x$ ; ie, the matching is pure.   Therefore, the stable matching $\gamma$ takes the form $\gamma=(Id, F_Y , F_Z)_\# \mu$, and as $F_Z$ is unique by part 1, and $F_Y$ is uniquely determined from \eqref{envelope for z} by $F_Z$, $\gamma$ is unique.

\end{proof}

\subsection{Fully mixed problems}

Our condition for purity and uniqueness will require a couple of definitions.  The first  is borrowed from \cite{KP2}.

\begin{definition}\textbf{(Splitting sets)}
For a fixed $x \in X$, a set $S_x \subset Y \times Z$ is a \emph{splitting set at $x$}  if there exist functions $V(y)$ and $W(z)$ such that 

$$
s(x,y,z) \leq V(y)+W(z)
$$
with equality whenever $(y,z) \in S_x$.
\end{definition}
The particular case when $W(z)=0$ in the above definition is especially relevant for this paper:
\begin{definition}\textbf{($z$-trivial splitting sets)}
For a fixed $x \in X$, a set $S_x \subset Y \times Z$ is a \emph{$z$-trivial splitting set at $x$}  if there exists a function $V(y)$ such that 

$$
s(x,y,z) \leq V(y)
$$
with equality whenever $(y,z) \in S_x$.
\end{definition}
It is clear that any $z$-trivial splitting set is a splitting set.
The role of $z$-trivial splitting sets in the matching problem \eqref{mm}  is  fairly transparent; for a given buyer $x$, the collection of all seller-contract pairs $(y,z)$ achieving equality in \eqref{stabilitycond}  (and hence potentially matching with $x$ in equilibrium) is a $z$-trivial splitting set at $x$.   As was observed in \cite{KP2}, splitting sets play a similar role in the tripartite  matching problem \eqref{classmm}.

\begin{remark}\label{maximality char}
It is worth noting that $S_x$ is a $z$-trivial splitting set at $x$ if and only if $\bar z$ maximizes  $z \mapsto s(x,\bar y,z)$ for each $(\bar y, \bar z) \in S_x$.
\end{remark}
\begin{definition}\textbf{(Twist on splitting sets)}
A differentiable surplus $s(x,y,z)$ is \emph{twisted on splitting sets} (or (TSS) for short), if whenever $S_x \subseteq  Y \times Z$ is a  splitting set at $x$ and $p \in \mathbb{R}^{n_x}$, there is at most one $(y,z) \in S_x$ such that 
$$
p= D_x s(x,y,z).
$$
\end{definition}
In \cite{KP2}, the (TSS) condition was shown to imply purity in the multi-agent matching model \eqref{classmm}.  Here, we introduce a variant, replacing splitting sets with $z$-trivial splitting sets, which will play an analagous role in \eqref{mm} and the related matching  problem.

\begin{definition}\textbf{(Twist on $z$-trivial splitting sets)}
A differentiable surplus $s(x,y,z)$ is \emph{twisted on $z$-trivial splitting sets} (or (TzSS) for short),  if whenever $S_x \subseteq  Y \times Z$ is a $z$-trivial splitting set at $x$ and $p \in \mathbb{R}^{n_x}$, there is at most one $(y,z) \in S_x$ such that
$$
p= D_x s(x,y,z).
$$
\end{definition}

We are now ready to state our main theoretical result on the purity of matchings.

\begin{theorem}\label{purity thm}
Suppose $s$ is twisted on $z$-trival splitting sets, and $\mu$ is absolutely continuous with respect to Lebesgue measure.  Then any stable matching $\gamma$ is pure.
\end{theorem}

The proof of this result is fairly standard; it involves applying the envelope theorem with respect to $x$ on the equality set in \eqref{stabilitycond} to equate the gradients of $U$ and $s$ (with respect to $x$), and then using the (TzSS) condition to infer the resulting equation can have only one solution.  There is one technical difficulty, which is also standard in problems of this type;  the payoff function $U(x)$  may not be differentiable.  We use a (well known) convexification trick to get around this, replacing $U$ with a Lipschitz (and hence differentiable almost everywhere, by Rademacher's theorem) payoff function, $\bar U$.  The proof can be found in the appendix.

A standard argument now implies uniqueness of the stable matching.
\begin{corollary}
Under the conditions in the preceding Theorem, the stable matching is unique.
\end{corollary}
\begin{proof}
Suppose $\gamma$ and $\bar \gamma$  are stable matchings; by Theorem \ref{purity thm}, we know that both $\gamma$ and $\bar \gamma$ are pure,  $\gamma =(Id, F)_{\#} \mu$ and  $\bar \gamma =(Id , \bar F)_{\#}  \mu$ for $F, \bar F:X \mapsto Y \times Z$, and, by Theorem \ref{varformulation}, both are also maximizers of \eqref{mm}.   It is then easy to see that $\gamma_{1/2} :=\frac{1}{2}[\gamma +\bar \gamma] \in \Gamma_{XYZ}(\mu,\nu)$.  It is therefore  also optimal in \eqref{mm}, as the functional is linear.  By Theorem \ref{purity thm} again, $\gamma_{1/2}$ too must then be supported on the graph of some function; on the other hand, it is clear that it is supported on the \emph{union} of the  graphs of $F$ and $\bar F$, which then implies that $F(x) =\bar F(x)$ almost everywhere, and so $\gamma =\bar \gamma$, yielding uniqueness. 
\end{proof}
As any $z$-trivial splitting set is a splitting set (one needs only to take $W(z) =0$ in the definition),  any surplus which is twisted on splitting sets is twisted on $z$-trivial splitting sets.  Therefore, we also have the following Corollary:

 \begin{corollary}\label{purity_uniqueness}
Suppose $s$ is twisted on  splitting sets, and $\mu$ is absolutely continuous with respect to Lebesgue measure.  Then the stable matching $\gamma$ is unique and pure.
\end{corollary}
The preceding Corollary is potentially useful, as several examples of surplus functions satisfying the twist on splitting sets condition are known, as well as general sufficient differential conditions ensuring it \cite{KP2}\cite{Pass14}.  Some of these will be discussed in Section \ref{sect: examples} below.

\section{Reformulation as a bipartite matching problem}
Here we provide a different, but equivalent, formulation of the problem, following Chiappori, McCann and Nesheim \cite{CMN}, as a binary matching, or two marginal optimal transport, problem. We define the \emph{reduced} surplus by:
\begin{equation}\label{redsur}
\bar s(x,y) =\max_{z \in Z} [u(x,y,z) +v(x,y,z)].
\end{equation}

The meaning of $\bar s(x,y)$ is clear; it expresses the maximum joint surplus (among all possible contracts) that can be generated by the partnership of $x$ and $y$.
The classical two marginal optimal transport problem is to maximize 

\begin{equation}\label{twomarg}
\int_{X \times Y} \bar s(x,y) d\sigma(x,y)
\end{equation}
over the set $\Gamma_{XY}(\mu,\nu)$ of probability measures on $X \times Y$  with $X$ (respectively $Y$) marginal $\mu$ (respectively $\nu$). 
This optimization problem is  equivalent to the classical strict stable matching problem under transferable utility, with surplus $\bar s$ \cite{GOZ}\cite{CMN}.

For each $x,y$, choose $\bar z(x,y) \in \argmax_z [u(x,y,z) +v(x,y,z)]$; note that $\bar z$ then defines a function $\bar z: X \times Y \rightarrow Z$.  Due to compactness, one can choose this  selection to be Borel measurable.

\begin{proposition}\label{prop: equivalence to reduced matching}
Suppose a measure $\sigma$ is optimal in \eqref{twomarg}.  Then $(Id, Id, \bar z)_{\#}\sigma$ is optimal for \eqref{mm}.  Conversely, if $\gamma$ is optimal in \eqref{mm}, then $\gamma_{XY}=(P_{XY})_{\#}\gamma$ is optimal in \eqref{twomarg}, where $P_{XY}(x,y,z) =(x,y)$.
\end{proposition}

The proof of this result is almost identical to the proof of the analagous result in \cite{CMN} and can be found in the appendix.

As the well known Spence-Mirlees condition on $\bar s$ is known to imply purity and uniqueness of  maximizers in \eqref{twomarg}, it is natural to look for conditions on $s$ which ensure it.  We show that the twist on $z$-trivial sets for $s$ is equivalent to the classical generalized Spence-Mirrlees condition on $\bar s$ (under an extra condition on $s$). Note that this, combined with the preceding proposition and Theorem \ref{matching_purity} yields an alternative proof of the buyer-seller aspects of the purity and uniqueness results in the last section (that is, buyer-seller purity and uniqueness of  $\gamma_{XY}$).


\begin{theorem}\label{equivalence to classcial twist}
Assume that both $s$ and $\bar s$ are everywhere differentiable with respect to $x$. If $s$ satisfies the twist  on $z$-trivial splitting sets condition, then $\bar s$ satisfies the twist condition.  

 Conversely, assume that $s$ is $x-z$ twisted.  Then, if $\bar s$ satisfies the twist condition,  $s$ satisfies the twist on $z$-trivial splitting sets condition.
\end{theorem}

The proof is relegated to an appendix. 

\begin{remark}
From inspection of the proof, it is clear that in fact slightly more is true.

If we assume that $\bar s$ is twisted, but remove the $x-z$ twist assumption on s, the argument in the proof of the second implication still yields that if $(y_0,z_0)$ and $(y_1,z_1)$ are in any splitting set $S_x$ at $x$, and $D_xs(x,y_0,z_0)= D_xs(x,y_1,z_1)$ then $y_0=y_1$ (although possibly $z_0 \neq z_1$). This then implies that any stable matching is buyer-seller pure, and that its projection onto $X \times Y$ is uniquely determined (as can also be proven using the twistedness of $\bar s$ in combination with Theorem \ref{matching_purity}).
\end{remark}

\begin{remark}
If $s$ is either Lipschitz or semi-convex, one can show by a standard argument that $\bar s$ is also Lipschitz or semi-convex, respectively, and it is well known that functions satisfying either one of these criteria are differentiable almost everywhere.  In fact, a version of the twist condition implying purity and uniqueness can be formulated under either of these assumptions (in place of everywhere differentiability) \cite{CMN}, and our proof in the appendix adapts easily to this setting.  It follows that one can remove the assumption of differentiability on $\bar s$ in the preceding theorem; we present the version with the differentiability assumption on $\bar s$  here for simplicity.
\end{remark}





\section{Examples}\label{sect: examples}

While the twist on $z$-trivial splitting sets condition looks complicated, it is possible to verify it on several classes of examples.  We present here three types of examples:
\begin{enumerate}
\item Examples satisfying the more restrictive twist on splitting sets condition (and hence the twist on $z$-trivial splitting sets condition introduced here as well). A wide variety of examples of this type are already known in the mathematical literature.
\item Examples violating the twist on splitting sets condition, but satisfying  twist on  $z$-trivial splitting sets.
\item An example violating twist on $z$-trivial splitting sets, together with an explicit non-pure stable matching.
\end{enumerate}


\subsection{Surpluses satisfying twist on splitting sets}
As mentioned above, a variety of examples satisfying the twist on splitting sets condition (and therefore also the twist on $z$-trivial splitting sets condition), as well as general differential conditions on $s$ which imply them, are known \cite{KP2}.  As the differential conditions are somewhat complicated, we do not state them here; instead, we present a couple of examples which seem potentially relevant in economics
\begin{example}\textbf{(One dimensional problems)}\label{1-d examples}

 Suppose $X,Y,Z \subset \mathbb{R}$ are all real intervals.  Then $s$ is twisted on splitting sets provided the \emph{compatibility} condition, $\frac{\partial^2s }{\partial x \partial y}[\frac{\partial^2s }{\partial z \partial y}]^{-1}\frac{\partial^2s }{\partial z \partial x}>0$, holds  for all $(x,y,z) \in X \times Y \times Z$.  In particular, this holds when $s$ is supermodular in each pair of its arguments.
\end{example}
The next example is similar to the Tinbergen model \cite{Tinbergen56}, augmented to include direct buyer-seller interactions.  

\begin{example}\textbf{(Bilinear utilities)}\label{bilinear_utilities}

  Suppose $X,Y,Z \subset \mathbb{R}^n$ are convex and 
$$
s(x,y,z) = x^TAy +x^TBz + y^TCz +f(x) +g(y)+h(z)
$$ 
for nonsingular $n \times n$ matrices $A,B$ and $C$.  Then $s$ is twisted on splitting sets provided the symmetric matrix $C^TA^{-1}B +B^T(A^T)^{-1}C$ is positive definite.
\end{example}
Note that the positive definiteness assumption on $C^TA^{-1}B +B^T(A^T)^{-1}C$ forces each of the matrices $A,B$ and $C$ to be invertible.
Proofs of the (TSS) property for the surplus functions in both of the examples in this subsection can be found in \cite{KP2}.

\begin{remark}  As we will see below, the sufficient conditions for purity and uniqueness considered here  (twist on  splitting sets) are substantially  stronger than the twist on $z$-trivial splitting sets, and so, when studying purity and uniqueness in the hybrid matching-hedonic model, the motivation for considering the multi-marginal coupling between buyers, sellers and (prescribed) goods and the related twist on splitting sets condition may seem questionable.  However, there are at least two concrete advantages to doing so.  

First, the twist on splitting sets condition is often easier to check.  For instance, in one dimension, the compatibility condition in Example \ref{1-d examples} is essentially equivalent to twist on splitting sets and hence is an easy  to check sufficient condition for twist on $z$-trivial splitting sets; if compatibility fails, twist on $z$-trivial splitting sets may  still hold, but establishing this typically requires more delicate arguments.

Secondly, the twist on splitting sets condition has some flexibility not shared by the twist on $z$-trivial splitting sets; namely, if $s(x,y,z)$ is twisted on splitting sets, then so is $s(x,y,z) +F(z)$ for any function $z$.  This fact may make it easier to check on certain examples.  
\end{remark}

\subsection{Surplus satisfying the twist on $z$-trivial splitting sets (but violating twist on splitting sets)}
The (TzSS) condition is strictly weaker than the (TSS) condition.  We demonstrate this here by presenting two examples which do not satisfy the twist on splitting sets condition, but do satisfy the weaker variant, twist on $z$-trivial splitting sets.
\begin{example}\textbf{(Strictly hedonic utilities)}\label{strictly hedonic}
Assume that the utilities of both consumers and producers depend only on goods, $u(x,y,z) =u(x,z)$ and $v(x,y,z) =v(y,z)$, and that for each fixed $x$ and $y$, all maxima of the function $z \mapsto u(x,z) +v(y,z)$ occur on the interior of $Z$.  Then $x-z$ twistedness on $u$ and $z-y$ twistedness on $v$ suffice to ensure the twist on $z$-trivial splitting sets condition on $s(x,y,z) =u(x,z) +v(y,z)$.
\end{example}

The proof of this assertion can be found in the appendix.

\begin{remark}
The conditions in Example \ref{strictly hedonic} \emph{do not} imply the twist on splitting sets condition, and as a result it is possible for the solution to the tripartite matching problem \eqref{classmm} with surplus $s(x,y,z)=u(x,z) +v(y,z)$ to be non-unique and non-pure.   Suppose, for example, $\alpha =\delta_{z_0}$ is concentrated at a point.  Then if a probability measure $\gamma$ on $X \times Y \times Z$ is in
$\Gamma_{XYZ}(\mu,\nu,\alpha)$ (ie, has marginals $\mu,\nu$ and $\alpha$) we have $z=z_0$, $\gamma$ almost surely, so that
\begin{eqnarray*}
\int_{X \times Y \times Z} s(x,y,z) d\gamma(x,y,z) &=&\int_{X \times Y \times Z} [u(x,z) +v(y,z)] d\gamma(x,y,z)\\
& =&\int_{X \times Y \times Z} [u(x,z_0) +v(y,z_0)] d\gamma(x,y,z) \\
& =&\int_Xu(x,z_0) d\mu(x)+\int_Yv(y,z_0) d\nu(y).
\end{eqnarray*}
As the last expression does not depend on $\gamma$, \emph{any} $\gamma \in \Gamma_{XYZ}(\mu,\nu,\alpha)$ maximizes the total surplus and is therefore stable.
\end{remark}

We close this subsection by revisiting the Tinbergen \cite{Tinbergen56} type surplus functions from Example \ref{bilinear_utilities}.  We show that the twist on $z$-trivial splitting sets holds in much greater generality that the twist on splitting sets (although we specialize slightly here, by replacing the general function $h(z)$ with a concave quadratic $z^tDz$).  
\begin{example}\label{nonsingular bilinear}
Suppose $X,Y,Z \subseteq \mathbb{R}^n$ are convex, and let 
$$s(x,y,z)=x^TAy +x^TBz + y^TCz +z^TDz+f(x)+g(y),$$
 with $D+D^T <0$.  Then $s$ is twisted on $z$-trivial splitting sets provided $C$ is invertible and
$$
B-A(C^T)^{-1}(D+D^T)
$$
 is non-singular.
\end{example}
\begin{proof}
Given a $z$-trivial splitting set $S_x$ at $x$, we note that if $(y,z) \in S_x$, maximality of $s(x,y,\cdot)$ at $z$ (recall Remark \ref{maximality char}) implies 
$$
0=D_zs(x,y,z) =B^Tx+ C^Ty+(D+D^T)z,
$$
so $y = -(C^T)^{-1}[B^Tx+(D+D^T)z]$.  For a given $p$, we will show that only one point of the form $(y,z) =(-(C^T)^{-1}[B^Tx+(D+D^T)z],z)$ can satisfy the condition
$$
p=D_xs(x,y,z)=Ay+Bz+Df(x) = A(-(C^T)^{-1}[B^Tx+(D+D^T)z]) +Bz+Df(x). 
$$
Indeed, the mapping 
$$
z \mapsto A(-(C^T)^{-1}[B^Tx+(D+D^T)z]) +Bz+Df(x) =-A(C^T)^{-1}B^Tx +[B-A(C^T)^{-1}(D+D^T)]z
+Df(x)$$ 
is affine and injective by assumption.  This completes the proof.
\end{proof}
\begin{remark}
In the model above, if $A$ is also invertible, we have 
$$
B-A(C^T)^{-1}(D+D^T) = A(C^T)^{-1}[C^TA^{-1}B-(D+D^T)].
$$
If in addition, we have $C^TA^{-1}B +B^T(A^T)^{-1}C>0$, then the surplus is twisted on splitting sets, according to Example \ref{bilinear_utilities}.  Twist on $z$-trivial splitting sets is much weaker, requiring only \emph{invertibility} of $[C^TA^{-1}B-(D+D^T)]$ rather than \emph{positivity} of its symmetric part, which is implied by the condition $C^TA^{-1}B +B^T(A^T)^{-1}C>0$ in Example \ref{bilinear_utilities}.

Furthermore, the matrices $B$ and $A$ in this example are \emph{not required} to have full rank; in particular, this model incorporates low dimensional buyer seller interactions, where preferences of buyers/sellers for their partners are dependent on only some of their characteristics (for instance, if all the entries of $A$ are $0$ except the upper left hand corner $A_{11}$,  partners' preferences depend only on the first characteristics, $x_1$ and $y_1$).  In its general form, the model interpolates between the strictly  hedonic case, where $A=0$, and the case with strong, full dimensional interactions between buyers and seller, when $A$ has full rank.
\end{remark}
\subsection{A surplus violating twist on $z$-trivial splitting sets, and a non pure solution}
Here we exhibit an example of a surplus violating the twist on $z$-trivial splitting sets condition, and demonstrate explicitly that in this case, matching equilibria may not be pure.

We let $X,Y,Z$ be intervals in $\mathbb{R}$; the consumers' and sellers' surplus are given respectively by $u(x,y,z) = xy+xz$ and $v(x,y,z)=-yz-z^2/2$, so that $s(x,y,z) =xy+xz-yz-z^2/2$.  It seems reasonable to interpret this surplus economically as a toy model for the effects of ethical business practices.  The variable $x \in X$ will represent the income of a consumer and $z \in Z$ the quality of a good.    Firms will be differentiated according to a variable $y \in Y$ which we may think of as reflecting the ethicality of their business practices (as perceived by consumers); for example, firms with large values of $y$ may provide their workers with better working conditions.  Consumers' preferences then have two supermodular terms, reflecting separately their preferences to buy higher quality goods and to purchase them from more ethical businesses (the supermodularity of $xy$ may be interpreted as consumers with more disposable income having stronger preferences for ethically produced goods than their lower income counterparts).  Producers' preferences are independent of consumers, but their costs $yz+z^2/2$ include a quadratic term in good quality and also a supermodular term $yz$, meaning more ethical firms have higher marginal production costs (for instance, producing a higher quality good may take more hours of labour than a lower quality good -- the resulting difference in cost will be higher for a firm paying higher wages).

Now note that for $U(x) =\frac{x^2}{2}$ and $V(y) =\frac{y^2}{2}$, we have

\begin{equation}
s(x,y,z) - U(x)-V(y) =-\frac{|x-y-z|^2}{2} \leq 0
\end{equation}
with equality when $x=y+z$.  Then taking $\gamma$ to be uniform measure on the set $X\times Y \times Z \cap \{x=y+z\}$, we immediately get that $\gamma$ is a stable matching measure for its marginals $\mu =(P_X)_\# \gamma$ and $\nu =(P_Y)_\# \gamma$, with payoff funtions $U$ and $V$.  This matching is certainly not pure; each consumer $x$ is indifferent among a continuum of choices of producers $y$.

We note that when consumer $x$ and producer $y$ match together, they exchange product $z =x-y$, for price $p_{x,y,z} =u(x,y,z)-U(x) =x(y+z) -x^2/2 =x^2/2$  .  A $y$ varies,  increasing favourability of the firm $y$ to the consumer $x$ is exactly offset by the decreasing quality of the good $z=x-y$ they exchange, and the price remains constant.
\begin{remark}
By example \eqref{nonsingular bilinear}, the surplus function $s(x,y,z) =xy+xz-yz-az^2$ \emph{is} twisted on $z$-trivial splitting sets for any constant $a$ \emph{other} than $a=\frac{1}{2}$, indicating that the previous example is highly non generic.
\end{remark}

\section{Conclusion}
This paper studies a general hybrid matching-hedonic model where agents match according to their preferences for both their partners and the good or contract they exchange.  In contrast to strict matching and strict hedonic problems, these mixed models do not seem to have received much theoretical attention yet, but are quite natural in a variety of settings. 

The hybrid problem has a natural connection with tripartite matching, or multi-marginal optimal transport; specifically, every stable matching in the hybrid model solves a corresponding optimal transport problem.  This observation, together with known results on multi-marginal optimal transport, can be exploited to reveal information on the structure of matching patters.  In particular, locally, the dimension of the support of a stable matching measure is controlled in terms of the mixed second order partial derivatives of the surplus; this result holds without any conditions on the distributions $\mu$ and $\nu$ of agents.    In addition, if $\mu$ is absolutely continuous, the twist on splitting sets condition is known to imply purity and uniqueness of stable matchings in multi-marginal optimal transport and therefore immediately implies the same for the hybrid problem.  It can also be used as a guide to develop a weaker variant, twist on $z$-trivial splitting sets, which implies purity and uniqueness in the hedonic-matching problem but not in the more general tripartite matching problem.

\begin{appendices} 
\section{Proofs}

\subsection{Proof of variational formulation: Theorem \ref{varformulation}}
The proof requires the following Lemma, expressing a duality result for the linear maximization \eqref{mm}.  
\begin{lemma}
\begin{equation}\label{duality}
\sup_{\gamma \in \Gamma_{XYZ}(\mu,\nu)} \int_{X \times Y \times X}s(x,y,z) d\gamma(x,y,z)=\inf_{U,V}\int_X U(x)d\mu(x) +\int_YV(y)d\nu(y)
\end{equation}
where the infimum on the right hand side is taken over the set of continuous functions $U \in C(X)$ and $V \in C(Y)$ satisfying $U(x) +V(y) \geq s(x,y,z)$ throughout $X \times Y \times Z$. Furthermore, the infimum on the right hand side is attained.
\end{lemma}
We will refer to the minimization on the right hand side as the dual problem to \eqref{mm}.  The lemma is a variant of the standard, optimal transport duality theorem, and it's proof is a straightforward adaptation of the proof of that result in \cite[Theorem 1.3]{V}. 
\begin{proof}
The Riesz representation theorem implies that the dual of $C(X\times  Y \times Z)$ is the set $M( X\times  Y \times Z)$ of signed regular Borel measures on $X \times  Y \times  Z$.  We define the functionals $F$ and $G$ on $C(X\times  Y\times  Z)$ by

\[ F(f)=
\left\{
  \begin{array}{l l}
0 &  \text{if }f(x,y,z) \geq s(x,y,z) \text{ for all } (x,y,z)\\
    +\infty &  \text{otherwise}
  \end{array} \right. \]

and

\[ G(f)=
\left\{
  \begin{array}{l l}
\int_X U(x)\mu(x) + \int_YV(y) d\nu(y) &  \text{if } f(x,y,z) =U(x) +V(y)\\
    +\infty &  \text{otherwise.}
  \end{array} \right. \]

Fenchel-Rockafellar duality (see, for example, Theorem 1.9 in \cite{V}) then asserts that

\begin{equation}\label{Fenchel_Rockafellar}
\inf_{f \in C(X\times  Y \times Z)}[F(f) +G(f)] =\sup_{\gamma \in M( X\times  Y \times Z)}[-F^*(-\gamma) -G^*(\gamma)]
\end{equation}
where $F^*$ and $G^*$ are the Legendre-Fenchel transforms of $F$ and $G$, respectively.  It is easy to check that the infimum above coincides with the infimum in \eqref{duality}.  On the other hand, we compute

\begin{eqnarray*}
F^*(-\gamma) &:=&\sup_{f \in C(X\times  Y \times Z)} [-\int_{ X\times  Y \times Z}f(x,y,z)d\gamma -F(f)]\\
&=& \sup_{f \in C(X\times Y \times Z),\text{ }f \geq s}-\int_{X\times Y \times Z}f(x,y,z)d\gamma\\
&=&-\inf_{f \in C(X\times Y \times Z),\text{ }f \geq s}\int_{X\times Y \times Z}f(x,y,z)d\gamma.
\end{eqnarray*}
Now, if $\gamma$ is not a positive measure, the infimum above is clearly $-\infty$, while if it is a positive measure, the infimum is attained at $f=s$. So we have

\[ F^*(-\gamma)=
\left\{
  \begin{array}{l l}
-\int_{X \times Y\times Z}s(x,y,z)d\gamma(x,y,z) &  \text{if } \gamma \geq 0\\
    +\infty &  \text{otherwise.}
  \end{array} \right. \]

Similarly, 

\begin{eqnarray*}
G^*(\gamma) &:=&\sup_{f \in C(X\times Y \times Z)} [\int_{X\times Y \times Z}f(x,y,z)d\gamma -G(f)]\\
&=& \sup_{(U,V) \in C(X) \times C(Y)}\Big[\int_{X\times Y \times Z}[U(x) +V(y)]d\gamma(x,y,z) -\int_{X}U(x)d\mu(x) -\int_YV(y)d\nu(y)\Big]\\
&=& \sup_{(U,V) \in C(X) \times C(Y)}\Big[\int_{X}U(x)d(\gamma_X-\mu)(x) +\int_YV(y)d(\gamma_Y-\nu)(y)\Big],
\end{eqnarray*}
where $\gamma_X =(P_X)_\#\gamma$ and $\gamma_Y =(P_Y)_\#\gamma$ are the projections of $\gamma$ onto $X$ and $Y$, respectively.  The integrals inside the supremum are clearly $0$ for each choice of $U,V$ if $\gamma$ has $\mu$ and $\nu$ as its $X$ and $Y$ marginals, respectively, and so the supremum is $0$ in this case.  If the marginals of  $\gamma$ are not $\mu$ and $\nu$, then the supremum is clearly $+\infty$, so we have

\[ G^*(\gamma)=
\left\{
  \begin{array}{l l}
0 &  \text{if } \gamma\text{ has marginals } \mu \text{ and } \nu\\
    +\infty &  \text{otherwise.}
  \end{array} \right. \]

Noting that, if $\gamma$ is a signed measure, $\gamma \in \Gamma_{XYZ}(\mu,\nu)$ is equivalent to $\gamma$ being non-negative and having $\mu$ and $\nu$ as it $X$ and $Y$ marginals, it is then straightforward to see that the supremum in \eqref{Fenchel_Rockafellar} is exactly the supremum in \eqref{duality}.

To obtain existence in the dual problem, note that for any $U,V$ such that $U(x) +V(y) \geq s(x,y,z)$, we have
$$
U^s(y):=\sup_{x,z}s(x,y,z) -U(x) \leq V(y)
$$
and then
$$
U^{ss}(x):=\sup_{y,z}s(x,y,z) -U^s(y) \leq U(x).
$$ 
Now, as $s$ is assumed Lipschitz, $U^s$ and $U^{ss}$ are Lipschitz with the same constant $C$, by a now classical argument of McCann \cite[Lemma 2]{m3}.   By shifting $U^s \rightarrow U^s +a$ and $U^{ss} \rightarrow U^s -a$,  we may also assume that $U^s(\bar x) =0$ for some fixed $\bar x$; together with the Lipschitz condition and compactness, this implies that $|U^s|,|U^{ss}| \leq K$ for some fixed $K$.  

Noting that we have $U^{ss}(x) +U^s(y) \geq s(x,y,z)$, and $\int_XU^{ss}(x) d\mu(x)+\int_YU^s(y)d\nu(y) \leq  \int_XU(x) d\mu(x)+\int_YV(y)d\nu(y)$, we may take the minimization in the dual problem over functions which in addition to the constraint $U(x) +V(y) \geq s(x,y,z)$ are Lipschitz with uniform constant $C$ and bounded by the uniform constant $K$.   This set is compact with respect to uniform convergence, by the Arzela-Ascoli theorem, which implies existence of a minimizer.

\end{proof}
The preceding lemma can be used to prove Theorem \ref{varformulation}; the solutions $U$ and $V$ to the dual problem turn out to be exactly the payoff functions.
\begin{proof}

Given a stable matching $\bar \gamma$ on $X \times Y \times Z$, and associated payoff functions $U(x)$ and $V(y)$, we integrate the inequality \eqref{stabilitycond} against any other matching $\gamma \in \Gamma_{XYZ}(\mu,\nu)$ to obtain
\begin{equation}
\int_{X \times Y \times Z} s(x,y,z)d\gamma(x,y,z) \leq \int_{X \times Y\times Z} U(x) +V(y) d\gamma(x,y,z)  =\int_{X } U(x) d\mu(x) +\int_{Y }V(y) d\nu(y).  
\end{equation}
On the other hand, the stability of $\bar \gamma$ means that we have equality in \eqref{stabilitycond} $\bar \gamma$ almost everywhere, so we have equality in the preceding argument when $\gamma =\bar \gamma$.  This means the stable matching $\bar \gamma$ is optimal in \eqref{mm}.

On the other hand, if $\bar \gamma$ solves \eqref{mm}, let $U$ and $V$ be solution to the dual problem, guaranteed to exist by the lemma above.  Then, we have $s(x,y,z) \leq U(x) +V(y)$ everywhere by definition, and so
$$
\int_{X \times Y\times Z} s(x,y,z)d\bar \gamma(x,y,z) \leq \int_{X \times Y\times Z} U(x) +V(y) d\bar\gamma(x,y,z)  =\int_{X } U(x) d\mu(x) +\int_{Y }V(y) d\nu(y).  
$$
However, the duality lemma states that we actually have equality, $\int_{X \times Y \times Z} s(x,y,z)d\bar\gamma(x,y,z)   =\int_{X } U(x) d\mu(x) +\int_{Y }V(y) d\nu(y)$, which is possible only if $s(x,y,z) = U(x) +V(y)$, $\bar \gamma$ almost everywhere.  Thus, $U$ and $V$ are payoffs for $\bar \gamma$, and so $\bar \gamma$ is stable.
\end{proof}

\subsection{Proof the twist on $z$-trivial splitting sets implies purity: Theorem \ref{purity thm}}
\begin{proof}
Let $\gamma$ be a stable equilibrium and $U,V$ the corresponding payoff functions.  Let $S \subseteq X\times Y \times Z$ be the set where equality is attained in \eqref{stabilitycond}; as $spt(\gamma)\subseteq S$, it suffices to show that for $\mu$ almost all  $x$, the set $S_x:=\{(y,z): (x,y,z) \in S\}$ is a singleton.  Note that, as an immediate consequence of the definition, $S_x$ is a $z$-trivial splitting set.

We first set $V^s(x) =\sup_{(y,z) \in Y \times Z}s(x,y,z) -V(y)$.  It is known that the fact that $s$ is Lipschitz in $x$ implies that $V^s$ is in fact Lipschitz as well \cite{m3}, and hence differentiable Lebesgue almost everywhere by Rademacher's theorem.  In addition, for a fixed $x$,  \eqref{stabilitycond} implies $s(x,y,z) -V(y) \leq U(x)$ for every choice of $(y,z)$, and so taking supremum over $(y,z)$ yields
$$
V^s(x) \leq U(x).
$$
Therefore, for all $(x,y,z)$, we have the following string of inequalities
$$
 U(x) +V(y) \geq V^s(x) +V(y) \geq s(x,y,z).
$$
Furthermore, as the first and last terms are equal on $S$, we must have equality throughout this set; in particular,

$$
V^s(x) +V(y) =s(x,y,z)
$$
on $S$.  Now, for every $x$ at which $V^s$ is differentiable, and $y,z \in S_x$, the envelope theorem implies

\begin{equation}\label{eqn: envelope}
D  V^s(x) =D_x s(x,y,z).
\end{equation}

However, as $S_x$ is a $z$-trivial splitting set at $x$,  the $(TzSS)$  condition implies that this uniquely determines $y$ and $z$.  That is, the splitting set $S_x$ is a singleton.  This holds for each $x$ where $V^s$ is differentiable, which is Lebesgue almost every $x$, and hence $\mu$ almost every $x$, by the absolute continuity of $\mu$.   For each such $x$, we define $(F_Y(x),F_Z(x))$ to be the unique $(y,z)$  satisfying \eqref{eqn: envelope}; $\gamma$ is then concentrated on the graph of $(F_Y,F_Z)$, and is therefore pure.  This completes the proof.
\end{proof}

\subsection{Proof of equivalence between the hybrid matching-hedonic problem and the reduced matching problem: Proposition \ref{prop: equivalence to reduced matching}}
\begin{proof}
Let $\gamma$ be a maximizer in \eqref{mm} and set $\sigma =\gamma_{XY} =(P_{XY})_\#\gamma $.  Then clearly $\sigma \in \Gamma_{XY}(\mu,\nu)$; we will show that it maximizes \eqref{twomarg}.   For any other $\tilde \sigma \in \Gamma_{XY}(\mu,\nu)$, set $ \tilde \gamma =(Id, Id, \bar z)_\# \tilde \sigma$ and note that $\tilde \gamma \in \Gamma_{XYZ}(\mu,\nu)$.

We have

\begin{eqnarray*}
\int_{X \times Y}\bar s(x,y) d\sigma(x,y) &=&\int_{X \times Y }s(x,y, \bar z(x,y)) d\sigma(x,y) \\
&=&\int_{X \times Y \times Z }s(x,y, \bar z(x,y)) d\gamma(x,y, z),\text{ as }  \sigma =(P_{XY})_\#\gamma \\
& \geq &\int_{X \times Y \times Z }s(x,y, z) d\gamma(x,y, z),\\
& \geq &\int_{X \times Y \times Z }s(x,y, z) d\tilde \gamma(x,y, z),\text{ as } \gamma  \text{ maximizes }\eqref{mm}\\
&=&\int_{X \times Y }s(x,y, \bar z(x,y)) d\tilde \sigma(x,y)\text{ as }\tilde \gamma =(Id \times Id \times \bar s)_\# \tilde \sigma\\
&=&\int_{X \times Y}\bar s(x,y) d\tilde \sigma(x,y)
\end{eqnarray*}
As $\tilde \sigma \in \Gamma_{XY}(\mu,\nu)$ was arbitrary, it follows that $\sigma$ is optimal in \eqref{twomarg}.  

On the other hand, let $\sigma$ be any maximizer in \eqref{twomarg} and $\gamma =(Id , Id , \bar s)_\#  \sigma$.  It is clear that $\gamma \in \Gamma_{XYZ}(\mu,\nu)$; we need to show that it maximizes \eqref{mm}.  For any other $\tilde \gamma \in   \Gamma_{XYZ}(\mu,\nu)$, we set $\tilde \sigma = (P_{XY})_\#\tilde \gamma$ and observe $\tilde \sigma \in \Gamma_{XY}(\mu,\nu)$.  We then have, by reasoning similar to the above,

\begin{eqnarray*}
\int_{X \times Y \times Z }s(x,y, z) d\gamma(x,y, z)&=& \int_{X \times Y} s(x,y,\bar z(x,y)) d\sigma(x,y)\\
&=&\int_{X \times Y}\bar  s(x,y) d\sigma(x,y)\\
&\geq&\int_{X \times Y}\bar  s(x,y) d\tilde \sigma(x,y)\\
&=& \int_{X \times Y}   s(x,y,\bar z(x,y)) d\tilde \sigma(x,y)\\
&=& \int_{X \times Y \times Z}  s(x,y,\bar z(x,y)) d\tilde \gamma(x,y,z)\\
&\geq & \int_{X \times Y \times Z}  s(x,y,z) d\tilde \gamma(x,y,z).
\end{eqnarray*}
This  yields optimality of $\gamma$ in \eqref{mm} and completes the proof.
\end{proof}

\subsection{Proof of equivalence between twistedness of $\bar s$ and twistedness on $z$-trivial splitting sets of $s$: Theorem \ref{equivalence to classcial twist}}
\begin{proof}
First suppose $s$ is twisted on $z$ trival splitting sets.  Fix $x$, set $V(y)=\bar s(x,y)$ and

$$
S_x=\{(y,z): z\in \argmax(s(x,y,z))\}
$$

Then $S_x$ is a $z$-trivial splitting set for $s$ at $x$, with splitting function $V$, as for any $y,z$ we have by definition

$$
V(y)  = \bar s(x,y) \geq s(x,y,z)  
$$
with equality whenever $(y,z) \in S_x$.

Now, choose $y_0,y_1$ satisfying

\begin{equation}\label{first_order_equality}
D_x\bar s(x,y_0) =D_x\bar s(x,y_1);
\end{equation}
we want to show $y_0 = y_1$. We can choose $z_0 \in \argmax_z (s(x,y_0,z))$ and $ z_1\in  \argmax_z (s(x,y_1,z))$, so that $(y_0,z_0), (y_1,z_1) \in S_x$.  By the envelope condition, we have

$$
D_x \bar s(x,y_0) =D_xs(x,y_0,z_0)
$$  

and 

$$
D_x \bar s(x,y_1) =D_xs(x,y_1,z_1)
$$

Combined with \eqref{first_order_equality}, this implies that $D_xs(x,y_0,z_0) = D_xs(x,y_1,z_1) $.  As $(y_0,z_0)$ and $ (y_1,z_1)$ both belong to the $z$- trivial splitting set $S$, the twist on $z$-trivial splitting sets hypothesis  now implies $(y_0,z_0)=(y_1,z_1)$; in particular, $y_0 =y_1$ as desired.

 Conversely, assume $\bar s$ is twisted, and suppose that $(y_0,z_0), ( y_1, z_1) \in S_x$, where $S_x$ is a $z$-trivial  splitting set at $x$, such that
\begin{equation}\label{splitting set equality}
D_xs(x,y_0,z_0) =D_xs(x, y_1,  z_1);
\end{equation}
we need to show $(y_0,z_0) =(y_1,  z_1)$.  Let $V$ be the splitting function for $S_x$; then

$$
V(y_0) \geq s(x,y_0,z)
$$
for all $z$, with equality for $z=z_0$.  As the left hand side is independent of $z$, this tells us that $z_0 \in \argmax_z s(x,y_0,z)$ and so $\bar s(x,y_0) =s(x,y_0,z_0)$.  The envelope theorem then yields

$$
D_x \bar s(x,y_0) =D_xs(x,y_0,z_0).
$$

An identical argument implies 

$$
D_x \bar s(x,y_1) =D_xs(x,y_1,z_1)
$$
and so we have $D_x \bar s(x,y_0) =D_x \bar s(x,y_1)$.  The twist condition then gives us $y_0 =y_1=:y$.   It remains to verify that $z_0=z_1.$   To this end, note that \eqref{splitting set equality} now becomes

$$
D_xs(x,y,z_0) =D_xs(x, y,  z_1)
$$
and so the $x -z$ twist condition implies $z_0=z_1$, completing the proof.
\end{proof}

\subsection{Proof that strictly hedonic surpluses are twisted on $z$-trivial splitting sets: assertion in Example \ref{strictly hedonic}}
This  result actually follows by combining Theorem \ref{equivalence to classcial twist}  with a result in \cite{Pass13}, which asserts that the reduced surplus $\bar s$ corresponding to a strictly hedonic surplus $s$ is twisted (under the assumptions in Example \ref{strictly hedonic}); however, we feel it is enlightening to provide a direct proof as well. 
\begin{proof}
Let $S_x$ be a $z$-trivial splitting set at $x$, with splitting function $V(y)$.  Assume $D_xs(x,y_0,z_0) = D_xs(x, y_1,  z_1)$ for $(y_0,z_0), ( y_1,  z_1) \in S_x$; we need to show $(y_0,z_0)=( y_1, z_1)$.   Note that the form of $s$ implies
$$
 D_xu(x,z_0) =  D_xs(x,y_0,z_0) = D_xs(x,y_1, z_1) =D_xu(x,z_1);
$$
the $x - z$ twist condition on $u$ then implies that $z_0=z_1$.

It remains to show $y_0 =y_1$.  Now,  $ z \mapsto u(x,z) +v(y_0,z)- V(y_0)$ is maximized at $z=z_0$, and so its derivative vanishes there (note $z \in Z^0$ by assumption):

$$
D_zu(x,z_0)+D_zv(y_0,z_0)=0
$$
or 
$$
D_zu(x,z_0)=-D_zv(y_0,z_0)
$$
Similarly,

$$
D_zu(x, z_1)=-D_zv(y_1, z_1),
$$
which, as $z_0= z_1$, combines with the above to yield $D_zv(y_0,z_0)=D_zv( y_1, z_0)$.  The $z -y$ twistedness of $v$ then yields $y_0=y_1$, which completes the proof.
\end{proof}
\end{appendices}
\bibliographystyle{plain}
\bibliography{biblio}
\end{document}